\documentclass[aps,prl,reprint,amsmath,amssymb,superscriptaddress,showkeys]{revtex4-1}
\usepackage[breaklinks,colorlinks=true]{hyperref}
\usepackage[T1]{fontenc}
\usepackage[utf8]{inputenc}
\usepackage{mathtools}
\usepackage{amsmath,amsthm,amssymb}
\usepackage[dvipsnames]{xcolor}
\usepackage{bm}
\usepackage{dsfont}
\usepackage{lipsum}
\usepackage{calrsfs}
\DeclareMathAlphabet{\pazocal}{OMS}{zplm}{m}{n}

\usepackage{color}
\usepackage{hyperref}
\usepackage{bm}
\usepackage{xspace}
\allowdisplaybreaks
\usepackage{tikz}
\usetikzlibrary{matrix}

\definecolor{brickred}{rgb}{0.8, 0.25, 0.33}
\newcommand\myshade{85}
\colorlet{mylinkcolor}{BrickRed}
\colorlet{mycitecolor}{NavyBlue}
\colorlet{myurlcolor}{Aquamarine}

\hypersetup{
  linkcolor  = mylinkcolor!\myshade!black,
  citecolor  = mycitecolor!\myshade!black,
  urlcolor   = myurlcolor!\myshade!black,
  colorlinks = true,
}

\let\emptyset\varnothing

\newcommand{\uu}[2]{u_{#1}^{#2}}

\begin{document}

\title{Characterising high-order interdependence via entropic conjugation}

\author{Fernando E. Rosas}
\email{f.rosas@sussex.ac.uk}
\affiliation{Sussex AI and Sussex Centre for Consciousness Science, Department of Informatics, University of Sussex}
\affiliation{Centre for Complexity Science and Center for Psychedelic Research, Department of Brain Science, Imperial College London}
\affiliation{Center for Eudaimonia and Human Flourishing, University of Oxford}
\affiliation{Principles of Intelligent Behavior in Biological and Social Systems (PIBBSS)}

\author{Aaron Gutknecht}
\affiliation{Campus Institute for Dynamics of Biological Networks, Georg-August University Göttingen}

\author{Pedro A. M. Mediano}
\affiliation{Department of Computing, Imperial College London}
\affiliation{Division of Psychology and Language Sciences, University College London}

\author{Michael Gastpar}
\affiliation{School of Computer and Communication Sciences, École Polytechnique Fédérale de Lausanne}

\newtheorem{definition}{Definition}
\newtheorem{conjecture}{Conjecture}
\newtheorem{theorem}{Theorem}
\newtheorem{lemma}{Lemma}
\newtheorem{proposition}{Proposition}
\newtheorem{corollary}{Corollary}
\newtheorem{example}{Example}
\newtheorem{remark}{Remark}

\begin{abstract}

\noindent
High-order phenomena play crucial roles in many systems of interest, but their analysis is often highly nontrivial. 
There is a rich literature providing a number of alternative information-theoretic quantities capturing high-order phenomena, but their interpretation and relationship with each other is not well understood. 
The lack of principles unifying these quantities obscures the choice of tools for enabling specific type of analyses. 
Here we show how an \emph{entropic conjugation} provides a theoretically grounded principle to investigate the space of possible high-order quantities, clarifying the nature of the existent metrics while revealing gaps in the literature. This leads to identify novel notions of symmetry and skew-symmetry as key properties for guaranteeing a balanced account of high-order interdependencies and enabling broadly applicable analyses across physical systems.

\end{abstract}

\maketitle

Physical and biological systems often exhibit relationships between their parts that cannot be reduced to dependencies in subsets of them~\cite{battiston2022higher}. 
The study of these high-order interdependencies has lead to new insights in a wide range of physical systems~\cite{battiston2020networks,battiston2021physics,santoro2023higher}, and also in studies 
involving genetics~\cite{cang2020inferring,park2021higher} and neural systems (both biological~\cite{gatica2021high,luppi2022synergistic,herzog2022genuine,varley2023information,varley2023multivariate} and artificial~\cite{tax2017partial,proca2022synergistic,kaplanis2023learning}), to name a few. Overall, qualitatively different types of interdependence have been found to play complementary roles balancing needs for robustness and flexibility~\cite{varley2024evolving,luppi2024information}.

There are different approaches to quantify high-order phenomena~\cite{rosas2022disentangling}, among which we focus on information-theoretic metrics based on Shannon entropy. 
While there is a rich literature offering such metrics, their interpretation is highly non-trivial --- 
being unclear if these quantities are capturing the same effects or instead provide complementary perspectives. This lack of clarity makes it challenging for researchers to choose the right tools to carry out specific types of analyses, severely hindering the study of high-order phenomena.

Here we address this issue by introducing the notion of \emph{entropic conjugation}, which establishes a theoretically grounded principle to explore the space of possible high-order quantities. Our results show that the existent high-order metrics have a closer relationship than previously thought, while revealing gaps in the literature for characterising interactions involving more than 5 variables. Moreover, the notions of symmetry and skew-symmetry with respect to conjugation emerge as key guarantees for providing a balanced account of high-order interdependence, enabling analyses that can illuminate the high-order profile of a wide range of physical systems. 
The proofs of our results can be found in the Appendix.

\vspace{.4cm}
\paragraph{Measures of multivariate interdependence.}

Let's consider a system with a state is specified by the random vector $\bm X=(X_1,\dots,X_n)$ following a joint distribution $p_{\bm X}$ and marginal distributions $p_{X_i}$. The literature presents various metrics to assess the dependencies between parts of $\bm X$; here we focus on linear combinations of entropies of the form $\phi(\bm X)=\sum_{\bm a\subseteq I_n} \lambda_{\bm a}H(\bm X^{\bm a})$, with $I_n=\{1,\dots,n\}$, $\bm X^{\bm a}$ is a vector of variables whose indices are in $\bm a\subseteq I_n$, $H$ is Shannon's entropy, and $\lambda_{\bm a}$ are scalars. 
We require these metrics to satisfy two key properties:
\begin{itemize}
    \item[(i)] \textit{Labelling-symmetry}: $\phi(\bm X)$ is invariant to permutations among $X_1,\dots,X_n$.
    \item[(ii)] \textit{Dependency}: $\phi(\bm X) = 0$ if the variables are jointly independent (i.e. $p_{\bm X} = \prod_{k=1}^n p_{X_k}$).
\end{itemize}
Hence, property (i) 
guarantees that $\phi$ does not depend on how variables are named and (ii) 
that it only captures interactions effects between variables~\footnote{There are multiple `directed' high-order quantities that distinguish between predictor and target variables, which don't satisfy labelling-symmetry~\cite{balduzzi2008integrated,timme2014synergy,rosas2024characterising}. These quantities will be studied in a follow-up work using a formalism that extends the one presented here.}.

There are several well-known metrics that satisfy these properties. The oldest of these is the \textit{interaction information}~\cite{mcgill1954multivariate}, which is defined as
\begin{align}
\text{II}(\bm X) 
\coloneqq&  \sum_{k=1}^n (-1)^{k+1} \sum_{|\bm a| = k} H(\bm{X}^{\bm a}).\label{eq:def_II}
\end{align}
Other metrics of interdependence are the \emph{total correlation} (TC)~\cite{watanabe1960information} and the \emph{dual total correlation} (DTC)~\cite{sun1975linear}, which are given by
\begin{align}
    \text{TC}(\bm X) \coloneqq& \sum_{j=1}^n H(X_j) - H(\bm X)\quad\text{and} \label{eq:def_TC}\\
    \text{DTC}(\bm X) \coloneqq& \: H(\bm X) -\sum_{j=1}^n H(X_j|\bm X^{-j}).\label{eq:def_DTC}
\end{align}
Another such metric, well-known in computational neuroscience, is the Tononi-Sporns-Edelman (TSE) complexity~\cite{tononi1994measure}, which is defined as
\begin{equation}
    \text{TSE}(\bm X) 
    := \sum_{k=1}^{\lfloor n/2\rfloor} {n\choose k}^{-1} \sum_{|\bm a|=k} I(\bm X^{\bm a}; \bm X^{-\bm a}),
\end{equation}
where $-\bm a$ is the set of indices of $\bm X$ that are not in $\bm a$.
Finally, we also consider the more recently introduced O-information and S-information~\cite{rosas2019quantifying}:
\begin{align}
    \Omega(\bm X) 
    :=& \,(n-2) H(\bm X) + \sum_{j=1}^n \Big( H(X_j) - H(\bm X^{-j}) \Big),\\
    \Sigma(\bm X)
    :=& \,\sum_{j=1}^n I(X_j;\bm X^{-j}).
\end{align}
Of these, TC, DTC, TSE, and $\Sigma$ are non-negative, while II and $\Omega$ can take positive and negative values. We will show these seemingly unrelated metrics can be parsimoniously unified under the concept of entropic conjugation.

\vspace{.5cm}
\paragraph{Characterising high-order interdependencies.}

Let us introduce the following average quantities:
\begin{equation}\label{eq:u_k}
      u_k(\bm X) :=  \frac{1}{{n\choose k+1} {k+1\choose 2}}
      \sum_{\substack{i,j\in I_n \\i<j}} \sum_{\substack{|\bm a|=k-1\\i,j\notin\bm a}} I(X_i;X_j | \bm X^{\bm a}) ,
\end{equation}
with $k=1,\dots,n-1$. 
These quantities satisfy labelling-symmetry and dependency, and capture the interdependencies between $k$ variables --- i.e. $u_j(\bm X)=0$ for $j<k$ if and only if all subsets of $k$ variables or less are statistically independent. Furthermore, it has been shown that all information-theoretic metrics $\phi$ satisfying labelling-symmetry and dependency can be expressed as
$\phi(\bm X) = \sum_{k=1}^{n-1} c_k u_{k}(\bm X)$, where $c_k\in\mathbb{R}$ captures the relevance of $(k+1)$-th order dependencies on $\phi$~\cite{te1978nonnegative}. 
Moreover, this decomposition is unique in guaranteeing that $\phi$ is non-negative if and only if $c_k\geq0$ for $k=1,\dots,n\!-\!1$.

A complementary, more fine-grained way of investigating high-order interdependence is enabled by partial information decomposition (PID), which addresses how information about a variable $Y$ provided by $\bm X$ may be decomposed into the contributions of its different components~\cite{williams2010nonnegative,wibral2017partial,mediano2021towards}. 
PID reveals that while pairwise interdependence is quantified by its strength (measured e.g. by the mutual information), higher-order relationships can be of qualitatively different kinds --- most notably redundant (multiple variables sharing the same information) or synergistic (a set of variables holding some information that cannot be seen from any subset). Moreover, PID recognises that synergy and redundancy can be mixed in non-trivial ways, and explores this thoroughly via an algebraic construction that leads to the decomposition
\begin{equation}
\label{eq:PID}
    I(\bm X;Y) = \sum_{\bm\alpha\in\mathcal{A}_n} I_\partial^{\bm\alpha}(\bm X;Y),
\end{equation}
where $\mathcal{A}_n$ is a collection of elements $\bm\alpha$ that cover all possible combinations of redundancy and synergy (App.~\ref{app:conjugation_PID}). For example, if $n=2$ then $\mathcal{A}_2$ has four elements: $\alpha_1\!=\!\{\{1\}\{2\}\}$ corresponding to the redundancy between $X_1$ and $X_2$, $\alpha_2\!=\!\{\{1,2\}\}$ corresponding to the synergy between them, and $\alpha_3\!=\!\{\{1\}\}$ and $\alpha_4\!=\!\{\{2\}\}$ corresponding to unique information in one but not the other.

\vspace{.5cm}
\paragraph{A conjugation of Shannon quantities.}

We are now ready to introduce the notion of entropy conjugation.

\begin{definition}\label{def:conj}
The \emph{entropic conjugation} is defined by
\begin{align}
\big(H(\bm X^{\bm a})\big)^* 
:=& H(\bm X^{-\bm a}) - H(\bm X).
\end{align}    
The conjugation of a linear combination of entropies $\phi(\bm X)=\sum_{\bm a\subseteq I_n} \lambda_{\bm a}H(\bm X^{\bm a})$ is
\begin{align}
\big(\phi(\bm X)\big)^* 
:=& \sum_{\bm a\subseteq I_n} \lambda_{\bm a}\big(H(\bm X^{\bm a})\big)^* .
\end{align}    
\end{definition}

It can be seen that $^*$ is a proper conjugation, as it is linear (by definition) and an involution, as $((H)^*)^* = H$. Also, a direct calculation shows that entropic conjugation acts on the mutual information as follows:
\begin{equation}
\big(I(\bm X^{\bm a};\bm X^{\bm b} | \bm X^{\bm c}) \big)^* 
= I\big(\bm X^{\bm a};\bm X^{\bm b}|\bm X^{-(\bm a\cup\bm b\cup\bm c)}\big),
\end{equation}
where $\bm a,\bm b,\bm c$ are disjoint subsets of indices. 
Furthermore, 
one can show that entropic conjugation exchanges high- for low-order interdependencies, which will be the basis of our analysis of high-order quantities in the next section.
\begin{proposition}\label{prop:u_conj}
    $\big(u_k(\bm X)\big)^* = u_{n-k}(\bm X).$
\end{proposition}

\noindent

A deeper insight on the effect of conjugation can be attained by looking at it via the PID framework. 
Our next result shows that entropic conjugation is the unique operation that arises from applying the duality principle from order theory~\cite{davey2002introduction} to PID, which results in a natural conjugation of atoms $\dag$ that exchanges redundancies for synergies and vice-versa (for example, if $\bm\alpha=\{\{1,2\}\}$ then 
$\bm\alpha^\dag=\{\{1\},\{2\}\}$). Crucially, this holds for any operationalisation of synergy and redundancy that is consistent with the PID framework (see App.~\ref{app:conjugation_PID}).

\begin{theorem}\label{prop:PID}
The natural conjugation of PID atoms $\dagger$ arising from order duality satisfies
\begin{equation}
I(\bm X^{\bm a};Y|\bm X^{\bm b})^\dagger 
= \!\!\!\sum_{\bm\alpha\in\mathcal{A}_{\bm a}^{\bm b}} \!\! I_\partial^{\bm\alpha^\dag}\!(\bm X;Y)
= I(\bm X^{\bm a};Y|\bm X^{\bm b})^*,
\end{equation}
where $\mathcal{A}_{\bm a}^{\bm b}$ is a suitable collection of atoms (see Lemma~\ref{prop:cmi_atoms_boolean}). 
\end{theorem}

Let's illustrate this result with a simple example. Using the fact that the O-information is equal to redundancy minus synergy~\cite{rosas2019quantifying}, Th.~\ref{prop:PID} implies that
\begin{align}
    \big( \Omega(X_1;X_2;Y)\big)^* \!
    &= I_\partial^{\{\{1\},\{2\}\}^\dag}\!(\bm X;Y) - I_\partial^{\{\{1,2\}\}^\dag}\!(\bm X;Y) \nonumber\\
    &= I_\partial^{\{\{1,2\}\}}(\bm X;Y)
    - I_\partial^{\{\{1\},\{2\}\}}(\bm X;Y) \nonumber\\
    &= - \Omega(X_1;X_2;Y).
\end{align}

\vspace{.5cm}
\paragraph{Symmetric and skew-symmetric metrics.}

We now use the entropic conjugation to introduce the notions of \emph{symmetric and skew-symmetric} interdepence quantities.
 
\begin{definition}\label{def:sym_skew}
    A linear combination of entropies $\phi$ is \emph{symmetric} if $(\phi)^*=\phi$ and \emph{skew-symmetric} if $(\phi)^*=-\phi$.
\end{definition}

This definition, combined with Prop.~\ref{prop:u_conj} and Th.~\ref{prop:PID}, implies that symmetric and skew-symmetric quantities provide balanced accounts of low- and high-order interdependencies (alternatively, redundancies and synergies): symmetric quantities weight these equally, while skew-symmetric weights them equally but with opposite signs. Thus, a practical way to recognise symmetric and skew-symmetric high-order metrics is via their weights in terms of the basis $u_k$, as shown next.
\begin{lemma}\label{lemma:char}
    If $\phi(\bm X)=\sum_{k=1}^{} c_k u_k(\bm X)$, then 
    \begin{itemize}
        \item $\phi$ is symmetric $\iff c_k=c_{n-k}$.
        \item $\phi$ is skew-symmetric $\iff c_k=-c_{n-k}$.
    \end{itemize}
\end{lemma}

With these tools at hand, we now study the existent high-order metrics under the lens of conjugation. 

\begin{proposition}\label{prop:decompositions}
The mentioned multivariate metrics can be decomposed as follows:
\begin{align}
    \normalfont{\text{TC}}(\bm X) 
    &= \sum_{k=1}^{n-1} (n-k)\uu{k}{n}(\bm X),\label{eq:TC_dec}\\
    \normalfont{\text{DTC}}(\bm X)
    &= \sum_{k=1}^{n-1} k\uu{k}{n}(\bm X),\label{eq:DTC_dec} \\
    \normalfont{\text{TSE}}(\bm X) 
    &= \sum_{k=1}^{n-1} \frac{k (n - k)}{2} u_k^n(\bm X),\\
    \Sigma(\bm X) 
    &= n \sum_{k=1}^{n-1} u_k^n(\bm X),\\
    \Omega(\bm X) 
    &= \sum_{k=1}^{n-1} (n-2k) u_k^n(\bm X),\\
    \normalfont{\text{II}}(\bm X)
    &= \sum_{k=1}^{n-1} (-1)^{k+1} {n-2\choose k-1} \uu{k}{n}(\bm X).
\end{align}
Therefore, the following relationships hold:
\begin{align}
\big(\normalfont{\text{TC}}(\bm X)\big)^* 
&= \normalfont{\text{DTC}}(\bm X),\\
\big(\Sigma(\bm X)\big)^* 
&= \Sigma(\bm X),\\
\big(\normalfont{\text{TSE}}(\bm X)\big)^* 
&= \normalfont{\text{TSE}}(\bm X),\\
\big(\Omega(\bm X)\big)^* 
&= -\Omega(\bm X),\\
\big(\normalfont{\text{II}}(\bm X)\big)^* 
&= (-1)^n \normalfont{\text{II}}(\bm X).   
\end{align}
\end{proposition}

These results show that the S-information and TSE complexity are balanced metrics of overall interdependence strength, while the O-information provides a balanced opposition between high- and low-order interdependencies. In contrast, the interaction information alternates between being symmetric or skew-symmetric in a way that will be better understood in the next subsection. Additionally, this result also shows that the TC and DTC are not balanced metrics, being duals to each other: the TC provides more weight to low-order effects while the DTC to high-order ones.

These results also reveal that this collection of metrics is not as arbitrary as it may seem: when seen from their coefficients $c_k$, they cover the constant (S-information), linear (TC, DTC, and O-information), quadratic (TSE), and binomial (interaction information) cases.

\vspace{.5cm}
\paragraph{Spanning the possible metrics.}
We now show that entropic conjugation induces a decomposition of high-order quantities into symmetric and skew-symmetric components --- revealing that skew-symmetric quantities are akin to the imaginary part of complex numbers.

\begin{theorem}\label{prop:decomp}
    Every information-theoretic metric of interdependence $\phi$ can be decomposed into unique symmetric and skew-symmetric components as follows:
    \begin{equation}
    \phi = \frac{1}{2}\underbrace{\big( \phi  + \phi^* \big)}_\text{symmetric} + \frac{1}{2}\underbrace{\big( \phi  - \phi^* \big)}_\text{skew-symmetric}.
    \end{equation}
    Moreover, symmetric and skew-symmetric components are orthogonal under the inner product induced by $\langle u_k,u_j\rangle  = \delta_j^k$.
\end{theorem}

This result provides a guide to investigate the geometry of the $(n-1)$-dimensional space of high-order metrics $\phi$ satisfying labelling-symmetry and dependency, which we denote by $\mathcal{I}_n$. 

\begin{corollary}\label{cor:signs_coeff_conj}
    If $\bm X$ has $n$ variables, then the $(n-1)$ dimensions of $\mathcal{I}(\bm X)$ are divided in the following way: 
    \begin{equation}
    \text{dim}\big(\mathcal{I}_n\big) = 
    \underbrace{\lfloor n/2 \rfloor}_\text{symmetric} + \underbrace{\lfloor (n-1)/2 \rfloor}_\text{skew-symmetric}.
    \end{equation}
\end{corollary}

\begin{figure*}
    \includegraphics[scale=1.65]{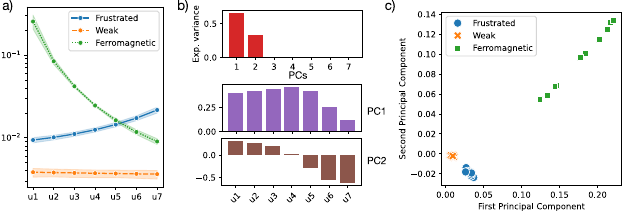}
    \caption{Information-theoretic analysis of the interdependencies observed in systems of $n=8$ spins subject to positive (ferromagnetic), negative (frustrated), and weak interactions between them. \textbf{a)} When calculating $u_k$, each type of interactions exhibit distinct profiles of interdependence. \textbf{b)} The variability among $u_k$ is captured by two principal components: one of symmetric character which accounts for the overall strength of the interdependence (PC1), and one of skew-symmetric character that accounts for the balance between high- and low-order interdependence (PC2). \textbf{c)} The values of $u_k$ projected onto these PCs provide a simple characterisation of these three types of systems in terms of their overall interdependence strength (PC1) and the balance between high- and low-order effects (PC2).}
    \label{fig1}
\end{figure*}

These results have the following consequences:
\begin{itemize}
    
    \item[-] For $n=2$ variables, Shannon's mutual information is the only symmetric functional (up to scaling), as there are no skew-symmetric functionals. 
       
    \item[-] For $n=3$ variables, the S-information is the only symmetric functional and the O-information (equivalently, the interaction information) is the only skew-symmetric one (up to scaling). Therefore, if $\phi\in\mathcal{I}_3$ then
        $\phi(\bm X) = \alpha \Sigma(\bm X) + \beta \Omega(\bm X)$.
    
    \item[-] For $n=4$ variables, the S-information and the interaction information span the subspace of symmetric metrics, while the O-information is the only skew-symmetric one (up to scaling). Therefore, if $\phi\in\mathcal{I}_4$ then
        $\phi(\bm X) = \alpha \Sigma(\bm X) + \alpha' \text{II}(\bm X) + \beta \Omega(\bm X)$.

    \item[-] For $n=5$ variables, the space of symmetric metrics is spanned by the S-information and the TSE-complexity, and the space of skew-symmetric metrics is spanned by the O-information and the interaction information. Therefore, if $\phi\in\mathcal{I}_5$ then
        $\phi(\bm X) = \alpha \Sigma(\bm X) + \alpha' \text{TSE}(\bm X) + \beta \Omega(\bm X) + \beta' \text{II}(\bm X)$.
\end{itemize}
Larger systems can be analysed in a similar fashion, but the existing metrics do not cover all the dimensions.

\vspace{.5cm}
\paragraph{Computational tractability.} 
Most $u_k$ (and, therefore, most high-order metrics) require estimating a large number of information-theoretic terms, and hence their computation becomes unfeasible when $n$ grows. 
Our next result characterises the space of possible computationally-efficient symmetric and skew-symmetric.
\begin{proposition}\label{prop:efficiency}
    The S-information and O-information are the only symmetric and skew-symmetric interdependence metrics that can be computed using a linear number of entropy terms. 
\end{proposition}
Note that the TC and DTC also require a linear number of terms, but they are neither symmetric or skew-symmetric --- in fact, their decomposition via Th.~\ref{prop:decomp}
yields $\text{TC}=(\Sigma+\Omega)/2$ and $\text{DTC}=(\Sigma-\Omega)/2$.

\vspace{.5cm}
\paragraph{Empirical results.}
To illustrate the applicability of this framework, we investigated the interdependencies exhibited by small spin systems under weak, ferromagnetic (positive), and frustrated (negative) types of interactions (App.~\ref{app:PCA}). The latter condition makes it impossible to simultaneously satisfy the tendency of all spins to be different from their neighbours, which is known for inducing high-order interdependencies~\cite{rosas2022disentangling}. 

To investigate the interdependencies of these systems, we calculated the values of $u_k$ according to Eq.~\eqref{eq:u_k} and identified the principal axes of variability by via principal component analysis (App.~\ref{app:PCA}). Results show that two components explain almost all the variability: a first component of symmetric character similar to the S-information, and a second component of skew-symmetric character similar to the O-information (Fig.~\ref{fig1}). In other words, an optimal information-theoretic analysis to characterise the high-order interdependence of these systems reduces to two keys aspects: (i) their strength, and (ii) the balance between high- and low-order components.

\vspace{.5cm}
\paragraph{Conclusion.} 
Here we investigated the space of possible metrics of high-order interdependence taking the form of linear combinations of Shannon entropies. 
We introduced the notion of \emph{entropic conjugation}, the effect of which can be understood in two complementary ways: as exchanging how metrics account for high- and low-order interdependencies, or alternatively, how they account for redundancies and synergies. Crucially, while multiple operationalisation of synergy and redundancy exist~\cite{wibral2017partial}, the properties of entropy conjugation hold for all approaches that are consistent with the PID formalism.

When studying high-oder quantities, non-negative metrics such as the S-information and TSE complexity were found to be invariant (i.e. symmetric) under entropic conjugation, confirming that they provide a balanced account of overall interdependence strength. 
Similarly, applying entropic conjugation to a signed metric such as the O-information results in a minus sing (i.e. skew-symmetric), guaranteeing that it provides a fair balance of the relative strength of redundancies and synergies. 
The interaction information was found to be either symmetric or skew-symmetric depending on the number of variables, providing a principled explanation to the observation (first made in Ref.~\cite{williams2010nonnegative}) that interpreting this quantity from a high-order perspective requires nuance.

This framework also let us prove that the well-known high-order metrics cover all possibilities when considering systems of up to $n=5$ variables, while the space of possible metrics for capturing interactions involving more variables remains largely unexplored. 
Additionally, the S-information and O-information were found to be the symmetric and skew-symmetric quantities that are most computationally efficient, and numerical analyses showed their relevance when studying physical systems.

\appendix

\section*{Appendix}

\section{Short proofs}

Here we present the proofs of our results. Some proofs use the quantities
\begin{equation}
    r_k(\bm X) = \frac{1}{{n\choose k}}\sum_{|\bm a|=k}
    H(\bm X^{\bm a}) \quad \text{for $k=0, 1, \ldots, n$,}
\end{equation}
with $r_0(\bm X)=0$. A direct calculation shows that
\begin{equation}
    u_k = 2r_k-r_{k+1}-r_{k-1}, \quad
    \text{for }
    k=1, 2, \ldots, n-1.
    \label{app:shortproofs:eq:ukrk}
\end{equation}
Moreover, we use the definition $\Delta^2 r_k := r_{k+1} + r_{k-1} - 2 r_k$ and the fact that $u_k = -\Delta^2 r_k$.

\begin{proof}[Proof of Prop.~\ref{prop:u_conj}]
Using the representation~\eqref{app:shortproofs:eq:ukrk}, a direct calculation shows that
\begin{align}
    (u_k)^* 
    &= (2r_k)^* - (r_{k+1})^* -(r_{k-1})^* \nonumber\\
    &= 2r_{n-k} - r_{n-k+1} -r_{n-k-1} \nonumber
    = u_{n-k}.
\end{align}
\end{proof}

\begin{proof}[Proof of Lemma~\ref{lemma:char}]
The proof follows directly from Definitions~\ref{def:conj} and \ref{def:sym_skew}.
\end{proof}

\begin{proof}[Proof of Prop.~\ref{prop:decompositions}]

The expressions for each metric can be directly verified by leveraging the representation~\eqref{app:shortproofs:eq:ukrk} and using the definition of $r_k$. The second part of the proposition follows then by using Lemma~\ref{lemma:char}.
\end{proof}

\begin{proof}[Proof of Th.~\ref{prop:decomp}]
    Let's denote as $S=(\phi + \phi^*)/2$ and $T=(\phi - \phi^*)/2$ the components proposed in the proposition, which can be directly shown to be symmetric and skew-symmetric and satisfying $\phi=S+T$. Let's assume there is another decomposition $\phi=S'+T'$ where $(S')^*=S'$ and $(T')^*=-T^*$. However this would imply that $\phi+\phi^* = 2S'$ and $\phi-\phi^*=-2T'$, which leads to $S=S'$ and $T=T'$, showing that the decomposition is unique. The orthogonality of these subspaces follows directly from Lemma~\ref{lemma:char}.
\end{proof}

\begin{proof}[Proof of Prop.~\ref{prop:efficiency}]
Consider $\phi=\sum_{k=1}^{n-1} c_k u_k$. Using the representation~\eqref{app:shortproofs:eq:ukrk}, one can re-write $\phi$ in terms of $r_k.$ Since every entropy term appears only in exactly one of the $r_k,$ there cannot be any term cancellations. Therefore, $\phi$ involves a linear number of (unconditional) entropy terms if and only if the resulting coefficients are non-zero only for $r_1$, $r_{n-1}$, and $r_n$.

Let us first consider the case in which $c_k=\alpha k + \beta$ is linear on $k$. Then, one can show that
\begin{align}
    \phi
    = - \sum_{k=1}^{n-1} c_k \Delta^2 r_k
    = -\alpha \sum_{k=1}^{n-1} k \Delta^2 r_k - \beta \sum_{k=1}^{n-1} \Delta^2 r_k.
\end{align}
By using the fact that $\Delta^2 r_k = \Delta r_{k+1} - \Delta r_{k}$ where $\Delta r_k := r_{k}-r_{k-1}$, one can use discrete calculus to find that
\begin{align}
    \sum_{k=1}^{n-1} \Delta^2 r_k
    &= \Delta r_n - \Delta r_1
    = r_n - r_{n-1} - r_1,\\
    \sum_{k=1}^{n-1} k \Delta^2 r_k 
    &=
    (n-1)\Delta r_{n}- \sum_{k=1}^{n-1} \Delta r_k^n\nonumber\\
    &= (n-1) r_n - n r_{n-1}.
\end{align}
This calculation leads to $\phi = \beta r^n_1 + (\alpha n+\beta)r_{n-1}^n - (\alpha(n-1)+ \beta) r_n^n$, showing that $\phi$ only includes a linear number of entropies.

To prove the converse statement, let's consider a quantity $\phi=\sum_k c_k u_k$, and let's denote its coefficients under $r_k$ and $\Delta r_k$ as $a_k$ and $b_k$ respectively, so that $\phi=\sum_i a_i r_i = \sum_j b_j \Delta r_j$ hold. As the $r_k$ are linearly independent, one can see that the above equations imply that the following conditions hold for all $k=1,\ldots,n-2$: 
\begin{equation}
a_k=b_k-b_{k+1}
\quad\text{and}\quad
b_j=c_j-c_{j+1}.
\end{equation}
Now, note that $a_k=0$ with $k\in \{2,\dots,n-2\}$ requires that $b_k=b_{k+1}$, and hence the above condition forces $b_1=\ldots=b_{n-1}$. Then, applying the same reasoning shows that $c_k-c_{k+1}$ has to be constant, which proves that $c_k$ depends linearly on $k$.

The above shows that the space of metrics that can be computed with a linear number of terms is two-dimensional, being spanned by DTC ($\alpha=1,\beta=0$) and the S-information ($\alpha=0,\beta=n$). This space is also spanned by the S-information and the O-information, concluding the proof.

\end{proof}

\vspace{-0.5cm}

\section{PID conjugation}
\label{app:conjugation_PID}

According to Eq.~\eqref{eq:PID}, PID introduces a decomposition of the mutual information $I(\bm X;Y)$ in terms of information atoms of the form $I_\partial^{\bm\alpha}(\bm X;Y)$, where $\bm\alpha=\{\alpha_1,\ldots,\alpha_l\}$ with $\alpha_j\subseteq I_n$ being sets of indices of the source variables $X_1,\dots,X_n$ such that no $\alpha_j$ contains another $\alpha_k$ --- making $\bm\alpha$ an `anti-chain' of sets of sources. Conceptually, $I_\partial^{\bm\alpha}(\bm X;Y)$ quantifies the information about the target variable $Y$ that is accessible via each collection of source variables $\alpha_1,\dots,\alpha_l$, while not being accessible via subsets of those collections or any other collections that not include them. For example, if $n=2$ then $I_\partial^{\{\{1,2\}\}}(\bm X;Y)$ corresponds to the information accessible in $(X_1,X_2)$ but not accessible from either $X_1$ or $X_2$ by themselves.

The accessibility relations denoted by PID antichains can be made explicit by `re-representing' them in terms of Boolean functions $f:\mathcal{P}_n \mapsto \{0,1\}$, where $\mathcal{P}_n$ is the powerset of $\{1, \dots, n\}$, taking a set of source indices as an input and returning $0$ or $1$ depending on whether the associated atom $I_\partial^f(\bm X;Y)$ is or isn't accessible via the set of sources. 
For example, the atom $\bm\alpha=\{\{1,2\}\}$ corresponds to the Boolean function that gives $f(\emptyset)=f(\{1\})=f(\{2\})=0$ and $f(\{1,2\})=1$. 
Crucially, it has been shown that there is a natural isomorphism between PID antichains and monotonic Boolean functions~\cite{gutknecht2021bits, gutknecht2023babel}, which implies that PID can be re-defined as follows.
\begin{definition}
A PID of the information provided by $\bm X= (X_1,\ldots, X_n)$ about $Y$ is a set of quantities $I_\partial^f(\bm X;Y)$ that satisfy for all $\bm a \subseteq \{1, \dots, n\}$
\begin{equation}
I(\bm X^{\bm a};Y) = \sum_{\substack{f\in\mathcal{B}_n \\ f(\bm a) = 1}} I_\partial^f(\bm X;Y),\label{eq:PID2}
\end{equation}
where $\mathcal{B}_n$ is the set of all non-constant monotonic Boolean functions $f:\mathcal{P}_n \mapsto \{0,1\}$.
\end{definition}

Note that Eq.~\eqref{eq:PID2} is equivalent to Eq.~\eqref{eq:PID}, with the only difference being the way in which PID atoms are labeled (either as antichains or Boolean functions). 
That said, viewing PID in terms of Boolean functions lets us conveniently handle various expressions, as shown below.

\begin{lemma}\label{prop:cmi_atoms_boolean}
Given two disjoint sets of source variables $\bm a$ and $\bm b$, we have
\begin{align}\label{def:PID_boolean}
I(\bm X^{\bm a};Y|\bm X^{\bm b}) 
&= \sum_{f\in\mathcal{B}_{\bm a}^{\bm b}}
I_\partial^f(\bm X;Y),
\end{align}
where $\mathcal{B}_{\bm a}^{\bm b} = \{ f\in\mathcal{B}_n: 
f(\bm a\cup\bm b)=1, f(\bm b)=0\}$. 
\end{lemma}

Note that the set $\mathcal{A}_{\bm a}^{\bm b}$ used in Th.~\ref{prop:PID} corresponds to the same atoms in $\mathcal{B}_{\bm a}^{\bm b}$ but represented in antichain form instead of as Boolean functions.

\begin{proof}
\begin{align}
I(\bm X^{\bm a};Y|\bm X^{\bm b}) &= I(\bm X^{\bm a},\bm X^{\bm b};Y) - I(\bm X^{\bm b};Y) \nonumber\\
&=\!\!\!\!\sum_{\substack{f\in\mathcal{B}_n\\
f(\bm a \cup \bm b)=1}} \!\!\!\!\! I_\partial^f(\bm X;Y) - \!\sum_{\substack{f\in\mathcal{B}_n\\
f(\bm b)=1}}\!\!I_\partial^f(\bm X;Y),
\end{align}
from which the desired result follows.
\end{proof}

The information atoms have a natural order in terms of their accessibility: an atom $I_\partial^f(\bm X;Y)$ can be said to be `more accessible' than an atom $I_\partial^g(\bm X;Y)$ if any set via which the latter is accessible is also a set via which the former is accessible. This property is elegantly captured by the Boolean function representation of atoms via the following partial ordering:
\begin{equation}
f \sqsubseteq g \text{ if and only if } 
f(\bm a) \leq g(\bm a)
\,\,\forall \bm a\in\mathcal{P}_n.
\end{equation}
Note that this partial ordering gives rise to a lattice of PID atoms denoted by $(\mathcal{B}_n,\sqsubseteq)$, which is isomorphic to the original PID lattice of antichains~\cite{gutknecht2023babel}.

Let us now introduce the notion of PID conjugation. For this, let us first note that, according to the order-theoretic principle of duality~\cite{davey2002introduction}, every lattice has a dual lattice in which all arrows are reverted. If we think of Boolean functions as bitstrings (with subsets ordered lexicographically), the dual of the PID lattice can be found by simply inverting the digits and reading the bitstring backwards, as shown by our next result.
\begin{proposition} 
\label{prop:pid_inv}
The mapping $f \mapsto f^\dagger$, where $f^\dagger$ is the Boolean function satisfying
\begin{equation}
\forall \bm a \subseteq \{1\ldots,n\}: f^\dagger(\bm a)=1 \Leftrightarrow f(\bm a^C)=0,
\end{equation}
is an order-reversing involution on $(\mathcal{B}_n,\sqsubseteq)$.
\end{proposition}
\begin{proof}
Involution: $f^{\dagger \dagger}(\bm a)=1$ if and only if $f((\bm a^C)^C)=f(\bm a)=1$. Hence, $f^{\dagger \dagger} = f$. Order-reversing: Suppose first that $f \sqsubseteq g$, meaning that $g(\bm a) = 1 \rightarrow f(\bm a) = 1$. Then, if $f^\dagger(\bm a) = 1$ we must have $f(\bm a^c)=0$ and so $g(\bm a^c) = 0$ and hence $g^\dagger(\bm a) = 1$.  Therefore we have $g^\dagger \sqsubseteq f^\dagger$. Conversely, if $g^\dagger \sqsubseteq f^\dagger$, then $f^\dagger(\bm a)=1 \rightarrow g^\dagger(\bm a)=1$. Hence, if $g(\bm a)=1$ it follows that $g^\dagger(\bm a^c)=0$ and thus $f^\dagger(\bm a^c)=0$ and thus $f(\bm a)=1$. Therefore,  $f\sqsubseteq g$.
\end{proof}

The effect of $\dag$ on the PID lattice can be understood as follows. Following Ref.~\cite{Jansma2024fast}, each atom can be expressed as concatenations of meet and join operations corresponding to redundancies and synergies. For example, the atom $\bm\alpha=\{\{1,2\},\{1,3\}\}$ can be constructed as $(1\vee 2)\wedge(2\vee3)$, where the join operation ($\vee$) can be thought of as denoting the union between sources (i.e., synergy), and the meet operation ($\wedge$) as the intersection between them (i.e., redundancy). Then, the involution introduced in Proposition~\ref{prop:pid_inv} can be understood as switching meets for joins and vice-versa. For example
\begin{align*}
    \{\{1,2\},\{1,3\}\}^\dagger 
    &= ((1\vee 2)\wedge(2\vee3))^\dagger \\
    &= (1\wedge2)\vee(1\wedge3) \\
    &\overset{(a)}{=} 1\vee(2\wedge3) = \{\{1\},\{2,3\}\},
\end{align*}
where equality $(a)$ uses the distributivity between meets and joins~\cite{Jansma2024fast}. 
This shows that the natural PID involution switches redundancy (i.e. easily accessible information, since it is contained in multiple sources) and synergy (i.e. difficult to access information, since one needs to observe multiple sources). Thus, the more easily an atom can be accessed, the more difficult it is to access its conjugate.

This PID involution leads to a natural conjugation between PID atoms, which we define next.
\begin{definition}\label{def:pid_conjugation}
The conjugate of a PID atom is given by
\begin{equation}
    \big(I_\partial^f(\bm X;Y) \big)^\dagger := I_\partial^{f^\dagger}(\bm X;Y),
\end{equation}
where $f^\dag$ is as defined in Prop.~\ref{prop:pid_inv}. 
Additionally, the conjugate of a linear combination of PID atoms $\psi(\bm X;Y) = \sum_{f \in \mathcal{B}_n} c_f I_\partial^f(\bm X;Y)$ is defined as
\begin{equation}
    \big(\psi(\bm X;Y) \big)^\dagger := 
    \sum_{f \in \mathcal{B}_n} c_f I_\partial^{f^\dagger}(\bm X;Y).
\end{equation}
\end{definition}

We now prove Proposition~\ref{prop:PID}, which states that applying PID conjugation on a conditional mutual information $I(\bm X^{\bm a};Y|\bm X^{\bm b})$ leads to the same outcome as entropic conjugation does --- namely, $I(\bm X^{\bm a};Y|\bm X^{(\bm a \cup \bm b)^C})$, where the complement is taken within the set of source variables.

\begin{proof}[Proof of Th.~\ref{prop:PID}] 
For $\bm a,\bm b$ disjoint subsets of $I_n$, then
\begin{align*}
I(\bm X^{\bm a};Y|\bm X^{\bm b})^\dagger &\overset{\text{lemma~\ref{prop:cmi_atoms_boolean}}}{=}  \left(\sum_{f(\mathbf{a}\cup \mathbf{b}) = 1 \,\&\, f(\mathbf{b})=0} \!\!\! I_\partial^f(\bm X;Y) \right)^\dagger \\
&\overset{\text{def. \ref{def:pid_conjugation}}}{=} \sum_{f(\mathbf{a}\cup \mathbf{b}) = 1 \,\&\, f(\mathbf{b})=0} I_\partial^{f^\dagger}(\bm X;Y) \\
&\overset{\text{def. $f^\dagger$}}{=} \sum_{f((\mathbf{a}\cup \mathbf{b})^C) = 0 \,\&\, f(\mathbf{b}^C)=1} I_\partial^f(\bm X;Y) \\
&= \sum_{f(\mathbf{a}\cup (\mathbf{a} \cup \mathbf{b})^C)=1 \,\&\, f((\mathbf{a}\cup \mathbf{b})^C) = 0} I_\partial^f(\bm X;Y) \\
&\overset{\text{lemma~\ref{prop:cmi_atoms_boolean}}}{=} I(\bm X^{\bm a};Y|\bm X^{(\bm a \cup \bm b)^C}),
\end{align*}
where the second to last equality follows because $\mathbf{a}\cup (\mathbf{a} \cup \bm b)^C = \bm b^C$ if $\bm a$ and $\bm b$ are disjoint. 
\end{proof}

Please note that the proof uses no properties specific to particular instantiations of synergy or redundancy, and hence the result holds for any operationalisation of these quantities that are consistent with the PID framework.

\section{Analysis of spin systems}
\label{app:PCA}

Our experiments considered systems of $n$ spins $\bm X=(X_1,\dots,X_n) \in \{-1,1\}^n$ following a Boltzmann distribution $p_{\bm X}(\bm x) = e^{-\beta H(\bm x)}/Z$ with a Hamiltonean of the form
\begin{equation}
    H(\bm x) = \frac{2}{n(n-1)}\sum_{i=1}^n\sum_{j=i+1}^n x_i x_k J_{i,k}~,
\end{equation}
where the coupling coefficients $J_{i,k}$ are i.i.d. sampled from a Gaussian distribution with mean $\mu$ and variance $\sigma^2$. The results reported in Fig.~\ref{fig1} corresponds to systems of $n=8$ spins with $\beta=1$, $\sigma^2=2$, and either $\mu=5$ (ferromagnetic), $\mu=0$ (weak), or $\mu=-5$ (frustrated). 

Our analysis pipeline was structured as follows. 
We first computed the joint distribution of $10$ systems of each type, and calculated the value of $u_k$ for each of them. 
The resulting values were then used to perform a principal component analysis. 
From this, we obtained the loadings of the two first principal components, denoted as $\xi_k$ and $\nu_j$, respectively. These loadings were then used to construct two high-order metrics, $\phi_{\text{PC}1}(\bm X) := \sum_{k=1}^{n-1} \xi_k u_k (\bm X)$ and $\phi_{\text{PC}2}(\bm X) := \sum_{j=1}^{n-1} \nu_j u_j (\bm X)$, which corresponds to projecting the value of the $u_k$'s onto the directions given by the principal components. Finally, these two resulting metrics were used to characterise the systems of interest.

\bibliography{references}
\bibliographystyle{apsrev4-1}

\end{document}